\documentclass[runningheads]{llncs}

\usepackage{graphicx}
\usepackage{todonotes}
\usepackage{amssymb}
\usepackage{amsmath}
\usepackage{subfigure}
\usepackage{mathtools}

\usepackage{url}
\usepackage{multirow}
\newcommand{\vir}[1]{``#1''}
\def\alph{\textsf{alph}}
\def\sp{\textsf{span}}
\def\spf{\textsf{span}}
\def\lm{\textsf{lm}}
\def\calG{\mathcal{G}}
\def\calF{\mathcal{F}}
\def\calT{\mathcal{T}}

\begin{document}

\title{String Attractors and Infinite Words}
\author{
    Antonio Restivo 
    \and Giuseppe Romana
    \and Marinella Sciortino
}
\authorrunning{A. Restivo et al.}
\institute{Dipartimento di Matematica e Informatica\\
    Università di Palermo, Italy\\
    \email{\{antonio.restivo,giuseppe.romana,marinella.sciortino\}@unipa.it}
}

\maketitle         

\begin{abstract}
The notion of \emph{string attractor} has been introduced in \cite{KempaP18} in the context of Data Compression and it represents a set of positions of a finite word in which all of its factors can be \vir{attracted}. The smallest size $\gamma^*$ of a string attractor for a finite word is a lower bound for several repetitiveness measures associated with the most common compression schemes, including BWT-based and LZ-based compressors. The combinatorial properties of the measure $\gamma^*$ have been studied in \cite{MRRRS_TCS21}. Very recently, a complexity measure, called \emph{string attractor profile function}, has been introduced for infinite words, by evaluating $\gamma^*$ on each prefix.  Such a measure has been studied for automatic sequences and linearly recurrent infinite words \cite{schaeffer_shallit2021string}. In this paper, we study the relationship between such a complexity measure and other well-known combinatorial notions related to repetitiveness in the context of infinite words, such as the factor complexity and the recurrence. 
Furthermore, we introduce new string attractor-based complexity measures, in which the structure and the distribution of positions in a string attractor of the prefixes of infinite words are considered. We show that such measures provide a finer classification of some infinite families of words.
\end{abstract}


\section{Introduction}
Compressibility and repetitiveness are two fundamental aspects in processing huge text collections \cite{Navarro_cacm22}. In many application domains, massive and highly repetitive data needs to be stored, analysed and queried. The main purpose of compressed indexing data structures is to store the texts and the structures needed to handle them by requiring space close to the size of the compressed data \cite{Navarro21_1}. In such a context, finding good measures able to capture the repetitiveness of texts is strictly related to having effective parameters to evaluate the performance, both in terms of time and space, of such compressed data structures.  For this reason, some of the most widely used repetitiveness measures are associated with effective compression schemes. For instance, we recall the number $z$ of phrases in the LZ77 parsing and the number $r$ of equal-letter runs produced by the Burrows-Wheeler Transform \cite{Navarro21_2}.
In such a framework, Kempa and Prezza proposed in \cite{KempaP18} a repetitiveness measure that, instead of being associated with a specific compressor, is related to some combinatorial properties of the text with the aim of unifying existing compressor-based measures. A \emph{string attractor} $\Gamma$ for a text $w$ is a set of positions in $w$ such each factor of $w$ must have an occurrence crossing some position in $\Gamma$. Intuitively, the more repetitive the text, the lower the number of positions of its attractor. The measure $\gamma^*(w)$ is the size of a string attractor of smallest size for $w$. On the one hand, it has been proved that $\gamma^*$ is a lower bound to all other compressor-based repetitiveness measure, on the other it is NP-complete to find the smallest attractor size $\gamma^*$ for a given text $w$. Combinatorial properties of the measure $\gamma^*$ for finite words have been explored in \cite{MRRRS_TCS21}.

In Combinatorics on words, the notion of repetitiveness has been declined in several ways and under a variety of aspects.  For instance, the \emph{factor complexity function} $p_x$ of an infinite word $x$ is a function that counts, for any $n>0$, the number of distinct factors of length $n$. Intuitively, the lower the factor complexity, the more repetitive an infinite word is. That is, the most repetitive words one can think of are those obtained by repeating the same factor infinitely many times, i.e. \emph{periodic words},  for which factor complexity takes on a constant value definitively. Among aperiodic words, \emph{Sturmian words} are the infinite words with minimal factor complexity. 
An infinite word $x$ is \emph{recurrent} if each factor of $x$ occurs infinitely often. The \emph{recurrence function} $R_x$ for an infinite word $x$, gives for each $n$, the size of the smallest window containing each factor of $x$ of length $n$, whatever such a window is located in $x$. Intuitively, it is strictly related to the maximum gap between successive occurrences of any factor, when all factors of length $n$ are considered. If such a gap is finite for each $n$, then the word is called \emph{uniformly recurrent}. For the \emph{linearly recurrent words} such a gap grows at most linearly with $n$.

Very recently, a bridge between these two different approaches has been presented in \cite{schaeffer_shallit2021string}, where the \emph{string attractor profile function} $s_x$ has been introduced. It measures, for each $n$, the size of a string attractor of smallest size for the prefix of length $n$ of an infinite word $x$. The behaviour of $s_x$ has been studied when $x$ is linearly recurrent word or an automatic sequence, whose symbols can be defined through a finite automaton \cite{allouche_shallit_2003}.

In this paper, we explore the relationship between the string attractor profile function of an infinite word $x$ and the other combinatorial notions of repetitiveness. In particular, we prove that the values that $s_x$ assumes for infinitely many $n$ give an upper bound to the factor complexity. On the other hand, we face the problem of searching for the necessary conditions, in terms of repetitiveness combinatorial properties, for the string attractor profile function to take values bounded by a constant.  Moreover, we study the behavior of the string attractor profile function for infinite words that are fixed point of a morphism, which represent a mathematical mechanism to generate repetitive words. 

Another contribution of this paper is to introduce two new complexity measures based on the notion of string attractor, which allow to obtain a finer classification of some infinite families of words. More in detail, we define the \emph{string attractor span complexity} (denoted by $\sp_x$) and the \emph{string attractor leftmost complexity} (denoted by $\lm_x$) of an infinite word $x$, which are related for each $n>0$ to the distribution of the positions within a string attractor of the prefix of $x$ of length $n$. These measures make it possible to distinguish infinite words that are indistinguishable under the action of the string profile function. 
In addition to exploring the relations between such measures and the periodicity and recurrence properties of an infinite word, we consider the class of infinite words for which the span complexity takes on a constant value infinitely many times. 
This allows us to obtain a new characterization of Sturmian words that are the infinite words with span complexity function equal to 1 for infinitely many $n$. More in general, we prove that if the span complexity $\sp_x$ assumes a constant value for each $n>0$, the aperiodic infinite word $x$ is a quasi-Sturmian word. Quasi-Sturmian words represent the simplest generalization of Sturmian words in terms of factor complexity.


\section{Preliminaries}
Let $\Sigma=\{a_1, a_2, \ldots, a_\sigma\}$ be a finite alphabet.
We denote by $\Sigma^*$ the set of finite words over $\Sigma$.
An infinite word $x=x_1x_2\ldots$ is an infinite sequence of characters in $\Sigma$.
Given a finite word $w = w_1 w_2 \cdots w_n$, we denote with $|w| = n$ the length of the word.
The \emph{empty-word} $\varepsilon$ is the only string that verifies $|\varepsilon| = 0$.
The \emph{reverse} of a word $w$ is the word read from right to left, that is $w^R = w_n w_{n-1} \cdots w_1$.
A finite word $v$ is called \emph{factor} of a word $x$ (finite or infinite) if there exist $i,j>0$ such that $j-i+1=|v|$ and $x[i,j]=x_i x_{i+1}\cdots x_j = v$.
We assume that $x[i,j] = \varepsilon$ if $j<i$.
We denote by $F(x)$ the set of all factors of $x$. The word $u$ is a \emph{prefix} (resp. \emph{suffix}) of $x$ if $x=uy$ (resp. $x=yu$) for some word $y$.
A factor $u$ of $x$ is \emph{right special} if there exist $a,b \in \Sigma$ with $a \neq b$ such that both $ua$ and $ub$ are factors of $x$.


\paragraph{String attractor of a finite word}
A string attractor for a word $w$ is a set of positions in $w$ such that all distinct factors of $w$ have an occurrence \emph{crossing} at least one of the attractor's elements. More formally, a {\em string attractor} of a finite word $w \in \Sigma^n$ is a set of $\gamma$ positions $\Gamma = \{j_1, \ldots, j_\gamma \}$ such
that every factor $w[i, j]$ of $w$ has an occurrence $w[i',j'] = w[i,j]$ with $j_k \in \{i', i'+1, \ldots, j'\}$, for some $j_k \in \Gamma$.
We denote by $\gamma^*(w)$ the size of a smallest string attractor for $w$.
We denote by $\alph(w)$ the set of the characters of $\Sigma$ appearing in $w$, i.e. $F(w)\cap \Sigma$. It is easy to see that $\gamma^*(w)\geq |\alph(w)|$. 

\begin{example}
\label{ex:string_attractor}
Let $w=adcbaadcbadc$ be a word on the alphabet $\Sigma=\{a,b,c,d\}$. The set $\Gamma=\{1,4,6,8,11\}$ is a string attractor for $w$.
Note that the set $\Gamma'=\{4,6,8,11\}$ obtained from $\Gamma$ by removing the position $1$ is still a string attractor for $w$, since all the factors that cross position $1$ have a different occurrence that crosses a different position in $\Gamma$. The positions of $\Gamma'$ are underlined in $$w=adc\underline{b}a\underline{a}d\underline{c}ba\underline{d}c.$$ $\Gamma'$ is also a smallest string attractor since $|\Gamma'|=|\Sigma|$. Then  $\gamma^*(w)=4$.
Remark that the sets $\{3,4,5,11\}$ and $\{3,4,6,7,11\}$ are also string attractors for $w$. It is easy to verify that the set $\Delta=\{1,2,3,4\}$ is not a string attractor since, for instance, the factor $aa$ does not intersect any position in $\Delta$.
\end{example}


\paragraph{Factor complexity}
Let $x$ be an infinite word. The \emph{factor complexity function} $p_x$ of $x$ counts, for any positive integer $n$, all the distinct factors of $x$ of length $n$, i.e. $p_x(n)=|F(x)\cap \Sigma^n|$.


\paragraph{Periodicity}
Given a word $x$, a natural number $p > 0$ is called \emph{period} of $x$ if $x_i = x_j$ when $i \equiv j \mod p$. 
An infinite word $x$ is called \emph{ultimately periodic} if there exist $u\in \Sigma^*$ and $v \in \Sigma^+$ such that $x = uv^\omega$, i.e. $x$ is the concatenation of $u$ followed by infinite copies of a non-empty word $v$.
If $u = \varepsilon$, then $x$ is called \emph{periodic}. An infinite word is \emph{aperiodic} if it is not ultimately periodic.


\paragraph{Recurrence and appearance functions}
An infinite word $x$ is said to be \emph{recurrent} if every factor that occurs in $x$ occurs infinitely often in $x$. 
The \emph{recurrence function} $R_x(n)$ gives, for each $n$, the least integer $m$ (or $\infty$ if no such $m$ exists) such that every block of $m$ consecutive symbols in $x$ contains at least an occurrence of each factor of $x$ of length $n$.
An infinite word $x$ is \emph{uniformly recurrent} if $R_x(n)<\infty$ for each $n>0$. 
$R_x(n)-n+1$ is the maximum gap between successive occurrences of any factor, when all factors of length $n$ are considered. 
If $R_x(n)$ is linear, $x$ is called \emph{linearly recurrent}. 
It is easy to see that an ultimately periodic word $x = uv^\omega$ with $u\neq \varepsilon$ is not recurrent.
On the other hand, if $x$ is periodic (the case $u=\varepsilon$) then $x$ is linearly recurrent.
Therefore, a recurrent word is either aperiodic or periodic.
Given an infinite word $x$, $A_x(n)$ denotes the length of the shortest prefix containing all the factors of $x$ of length $n$. The function $A_x(n)$ is called \emph{appearance function} of $x$.

\begin{remark}
It is known that $A_x(n)\leq R_x(n)$ (see \cite{allouche_shallit_2003}). Moreover, for any infinite word $x$ and for each $n>0$, since $|\Sigma|$ is finite, $A_x(n)$ is always defined and $A_x(n)\geq p_x(n)+n-1 = \Omega(n)$.
\end{remark}


\paragraph{Power freeness} An infinite word $x$ is said \emph{$k$-power free}, for some $k>1$, if for every factor $w$ of $x$, $w^k$ is not a factor of $x$.
If for every factor $w$ of $x$ exists $k$ such that $w^k$ is not a factor of $x$, then $x$ is called {\em $\omega$-power free}.


\paragraph{Morphisms} They represent a very interesting way to generate an infinite family of words.
Let $\Sigma$ and $\Sigma'$ be alphabets.
A \emph{morphism} is a map $\varphi$ from $\Sigma^*$ to $\Sigma'^*$ that obeys the identity $\varphi(uv) = \varphi(u)\varphi(v)$ for all words $u, v \in \Sigma^*$.
A morphism $\varphi$ is called \emph{prolongable} on a letter $a\in \Sigma$ if $\varphi(a)=au$ with $u\in \Sigma^+$.
If for all $a\in \Sigma$ holds that $\varphi(a)\neq \varepsilon$, then the morphism $\varphi$ is called \emph{non-erasing}.
Given a non-erasing morphism $\varphi$ prolongable on some $a\in\Sigma$, the infinite family of finite words $\{a, \varphi(a), \ldots, \varphi^i(a), \ldots\}$ are prefixes of a unique infinite word $\varphi^{\infty}(a)=\lim_{i\mapsto\infty}\varphi^i(a)$, that is called \emph{purely morphic word} or \emph{fixed point of $\varphi$}.
A morphism $\varphi$ is called \emph{primitive} if exists $t>0$ such that $b \in F(\varphi^t(a))$, for every pair of symbols $a,b \in \Sigma$.
If exists $k$ such that $|\varphi(a)| = k$ for every $a \in \Sigma$, then the morphism is called \emph{$k$-uniform}.


\paragraph{String attractor profile function}
Let $x$ be an infinite word. For any $n>0$, we denote by $s_x(n)$ the size of a smallest string attractor for the prefix of $x$ of length $n$. The function $s_x$ is called \emph{string attractor profile function} of $x$.
Such a notion has been introduced in \cite{schaeffer_shallit2021string}.

\begin{example}
\label{ex:Thue-Morse}
Let us consider the Thue-Morse word $$t=0110100110010110\cdots,$$
that is the fixed point of the morphism $0 \mapsto 01$, $1 \mapsto 10$.
It has been proved in \cite{schaeffer_shallit2021string} (cf. also \cite{KutsukakeMNIBT20}) that $s_{t}(n)\leq 4$ for any $n>0$.
Moreover, it is known that the functions $p_t(n)$, $R_t(n)$ and $A_t(n)$ are $\Theta(n)$.
See \cite{allouche_shallit_2003} for details.
\end{example}


\section{String attractor profile function, factor complexity and recurrence}

In this section we explore the relationships among different functions that aim to measure the repetitiveness of factors within infinite sequences of symbols. 

The following theorem establishes a relationship among appearance, factor complexity and string attractor profile functions. In particular, it shows that upper bounds on $s_x$ can induce upper bounds on $p_x$.

\begin{theorem}
\label{th:boundr}
Let $x$ be an infinite word.  For all $n > 0$,   
one has  $$p_x(n)  \leq  n\cdot s_x(A_x(n)).$$
\end{theorem}
\begin{proof}
Let us consider the value $A_x(n)$ representing the length of the smallest prefix of $x$ containing all the factors of $x$ of length $n$. Since the alphabet is finite, the value $A_x(n)$ is finite. By definition $s_x(A_x(n))$ is the size of the smallest string attractor of the prefix of length $A_x(n)$. Therefore, each factor of $x$ of length $n$ crosses at least one element of the string attractor. Since each element of the string attractor is crossed by at most $n$ distinct factors of $x$ of length $n$, one has $p_x(n)\leq n\cdot s_x(A_x(n))$.
\qed
\end{proof}

From previous theorem, the following corollary can be deduced. 

\begin{corollary}
\label{co:bounded-profile}
Let $x$ be an infinite word.
If exists $k>0$ such that $s_x(n)<k$ for each $n>0$, then $p_x(n)\leq n\cdot k$.
\end{corollary}

In other words, Corollary \ref{co:bounded-profile} states that if an infinite word has the string attractor profile function bounded by some constant value, then it has at most linear factor complexity.
We know that, in general, the converse of Corollary \ref{co:bounded-profile} is not true.
In fact, there are infinite words $x$ such that the factor complexity is linear and $s_x(n)$ is not bounded.
For instance, in Example \ref{ex:ch_powers_2} we consider the  characteristic  sequence $c$ of the powers of $2$. 

\begin{example}
\label{ex:ch_powers_2}
Let us consider the characteristic sequence $c$ of the powers of $2$, i.e. $c_i = 1$ if $i=2^j$ for some $j\geq0$, $0$ otherwise.
$$c = 1101000100000001 \cdots.$$
It is easy to see that $c$ is aperiodic and not recurrent because the factor $11$ has just one occurrence.
It is known that $p_c(n)$ and $A_c(n)$ are $\Theta(n)$ (\cite{allouche_shallit_2003}).
One can prove that $s_c(n)=\Theta(\log n)$ (\cite{KociumakaNP20,MRRRS_TCS21,schaeffer_shallit2021string}).
\end{example}

We raise the following:

\begin{question}
\label{conj:sx_bounded}
Let $x$ be an  uniformly recurrent word such that $p_x$ is linear. Is $s_x(n)$ bounded by a constant value?
\end{question}

Remark that, by assuming a stronger hypothesis on the recurrence of $x$, a positive answer to Question \ref{conj:sx_bounded} can be given, as stated in the following theorem proved in \cite{schaeffer_shallit2021string}. Such a result can be applied to describe the behaviour of the string profile function $s_t(n)$ for the Thue-Morse word $t$, as shown in Example \ref{ex:Thue-Morse}.

\begin{theorem}[\cite{schaeffer_shallit2021string}]
\label{th:shallit-schaeffer_lin_rec}
Let $x$ be an infinite word.
If $x$ is linearly recurrent (i.e. $R_x(n)=\Theta(n)$), then $s_x(n)=\Theta(1)$.  
\end{theorem}

The following proposition shows that also in case of ultimately periodic words, the string attractor profile function assumes a constant value, definitely. 

\begin{proposition}
\label{prop:ult_per_sx_bounded}
Let $x$ be an infinite word.
If $x$ is ultimately periodic, then $s_x(n) =\Theta(1)$.
\end{proposition}

\begin{proof}
Let $u\in \Sigma^*$ and $v \in \Sigma^{+}$ such that $x = uv^{\omega}$.
Since every periodic word is linearly recurrent, if $u=\varepsilon$ by Theorem \ref{th:shallit-schaeffer_lin_rec} the thesis holds.
If $u\neq \varepsilon$, then for every $n > |u|$ we can use a bound on the size $\gamma^*$ with respect to the concatenation provided in \cite{MRRRS_TCS21}, which says that for any $u,v\in\Sigma^+$, it holds that $\gamma^*(uv) \leq \gamma^*(u) + \gamma^*(v) + 1$.
Therefore, $s_x(n) \leq \gamma^*(u) + s_{v^\omega}(n-|u|) + 1 \leq |u| + k' + 1$.
On the other hand, for all the prefixes of length $n\leq |u|$ it holds that $s_x(n) \leq |u| < |u| + k' +1$.
Since $|u|$ and $k'$ are constant, we can choose $k = |u|+k'+1$ and the thesis follows.
\qed
\end{proof}

An interesting upper bound on the function $s_x$ can be obtained by assuming that the appearance function is linear, as shown in \cite{schaeffer_shallit2021string} and reported in the  following theorem.

\begin{theorem}[\cite{schaeffer_shallit2021string}]
\label{th:shallit-schaeffer-lineare_appearence}
Let $x$ be an infinite word.
If $A_x(n)=\Theta(n)$, then $s_x(n)=O(\log n)$.   
\end{theorem}

On the other hand, if the function $s_x$ is bounded by some constant value, the property of power freeness can be deduced, as proved in the following theorem. 

\begin{proposition}
\label{prop:omega-powerFree}
Let $x$ be an infinite word.
If $s_x(n) = \Theta(1)$, then $x$ is either ultimately periodic or $\omega$-power free.
\end{proposition}
\begin{proof}
By Proposition \ref{prop:ult_per_sx_bounded}, the thesis holds for every ultimately periodic word.
So, let us assume $x$ is aperiodic.
By contraposition, suppose $x$ is not $\omega$-power free.
Then there exists a factor $w$ of $x$ such that, for every $q>0$, $w^q$ is factor of $x$.
Moreover, $x \neq uw^\omega$ for every $u\in \Sigma^*$, otherwise $x$ would be ultimately periodic.
It follows that we can write $x= v_0 \cdot \prod_{i=1}^\infty w^{q_i} v_i $, with $v_0 \in \Sigma^*$, and, for every $i\geq1$, $q_i > 0$ and $v_i \in \Sigma^+$ such that $v_i$ does not have $w$ neither as prefix nor as suffix. 
Observe that there exist infinitely many distinct factors of the form $v_j w^{q_j} v_{j+1}$ for some $j\geq0$ and for each of these distinct factors we have at least one position in the string attractor.
Thus, for every $k>0$ exists $n>0$ such that $s_x(n)>k$ and the thesis follows. 
\qed
\end{proof}

On the other hand, the converse of Proposition \ref{prop:omega-powerFree} is not true for $\omega$-power free words.
Such a result leads to the formulation of the following Question \ref{conj:sx_bounded_omega}. Note that a positive answer to Question \ref{conj:sx_bounded_omega} implies a positive answer to Question \ref{conj:sx_bounded}: 

\begin{question}
\label{conj:sx_bounded_omega}
Let $x$ be $\omega$-power free word such that $p_x$ is linear. Is $s_x(n)$ bounded by a constant value?
\end{question}

The following examples show that for many infinite words  known in literature the string attractor profile function is not bounded by a constant. So, it could be interesting to study its behaviour.
In particular, Example \ref{ex:v_increasing_0s} shows that there exist recurrent (not uniformly) infinite words $x$ such that the function $s_x$ is unbounded.
However, one can find a uniformly recurrent infinite word $t$ such that $s_t$ is unbounded, as shown in Example \ref{ex:Toeplitz} .  

\begin{example}
\label{ex:v_increasing_0s}
Let  $n_0, n_1, n_2, n_3,\ldots$ be a increasing sequence of positive integers. Let us define the following sequence of finite words: 
$v_0  =  1$, $v_{i+1}  =  v_i 0^{n_i} v_i$, for $i>0$. 
Let us consider $v=\lim_{i\to\infty }v_i$. It is possible to verify that $v$ is recurrent, but not uniformly, and $s_v$ is unbounded.
\end{example}

\begin{example}
\label{ex:Toeplitz}
Toeplitz words are infinite words constructed by an iterative process, specified by a Toeplitz pattern, which is a finite word over the
alphabet $\Sigma\cup \{?\}$, where $?$ is a distinguished symbol not belonging to $\Sigma$ \cite{CassaigneK97}. Let us consider the alphabet $\Sigma=\{1,2\}$ and the pattern $p=12???$. The Toeplitz word $z$ (also called $(5,3)$-Toeplitz word) is generated by the pattern $p$ by starting from the infinite word $p^\omega$, obtained by repeating $p$ infinitely. Next, the location of $?$ are replaced by $p^\omega$, and so forth. So,  $$z=121211221112221121121222112121121211222212112\cdots.$$
It is known that all Toeplitz words are uniformly recurrent and, as shown in \cite{CassaigneK97}, $p_z(n)=\Theta(n^r)$ with $r=(\log5)/(\log5-\log3)\approx3.15066$. By applying Corollary \ref{co:bounded-profile}, we can deduce that $s_z$ is unbounded.
\end{example}
On the other hand, in support of the fact that $s_x(n)$ can be bounded by a constant value by using weaker assumptions than those of Theorem \ref{th:shallit-schaeffer_lin_rec}, we can show there exist uniformly (and not linearly) recurrent words for which $s_x(n)$ is bounded. We can consider the infinite word $u$ described in Example \ref{ex:holub} such that $p_u$ is linear and $s_u$ is constant. A larger class of examples is represented by some Sturmian words, as shown in Section \ref{sec:constant_span}.

\begin{example}\label{ex:holub}
Let us consider the following infinite word $u$ introduced by Holub in \cite{Holub14} and defined as follows. Let $\{n_i\}_{i\geq 1}$, be an increasing sequence of positive integers with $n_1\geq 2$. Then we define inductively the sequence $u_i$, $i\geq 0$, as $u_0=\epsilon$, $u_i=u_{i-1}a(u_{i-1}b)^{n_i}u_{i-1}$.
Let us consider $u=\lim_{i\to\infty}u_i$. 
It has been proved in \cite{Holub14} that $u$ is uniformly recurrent but not linearly recurrent. Moreover, for each $i\geq 1$, $u$ can be factorized as a product of words $u_ia$ and $u_ib$, i.e.  $u=u_ic_1u_ic_2u_ic_3\cdots$, where $c_j\in\{a,b\}$. More precisely, each occurrence of $u_i$ starts at position that is a multiple of $|u_i|+1$. 
By using the above properties, we can prove that $p_u(n)=2n$. Furthermore, is possible to prove that, for $i\geq 0$, the set $$\{|u_i|+1,\sum_{k=0}^{i}(|u_k|+1),|u_{i+1}|-|u_i|\}$$ is a string attractor for $u_{i+1}$ and a string attractor of constant size can be deduced for each prefix of $u_{i+1}$. Hence, $s_u(n)$ is $\Theta(1)$. 
\end{example}

All the infinite words considered in the paper, with information on string attractor profile function, factor complexity and recurrence properties, are summarized in Figure \ref{fig:my_label}. 

\begin{figure}[ht]
\begin{center}
    $
    \begin{array}{{r}||{c}|{c}|{c}}
    \text{Infinite word } x & p_x(n) & \text{Recurrence} & s_x(n) \\
    \hline
    \hline
    \text{Period-doubling word } p  \text{ (Ex. \ref{ex:period_doubling})}& \Theta(n) & \text{Linearly recurrent} & 2 \\[1.5pt]
    \text{Thue-Morse word } t \text{ (Ex. \ref{ex:Thue-Morse})}& \Theta(n) & \text{Linearly recurrent} & 4 \\[1.5pt]
    \text{Holub word } u \text{ (Ex. \ref{ex:holub})}& \Theta(n) & \text{Uniformly recurrent} & 3 \\[1.5pt]
    \text{Charact. Sturmian word } s \text{ (Thm. \ref{th-standard_profile})}& \Theta(n) & \text{Uniformly recurrent} & 2 \\[1.5pt]
    \text{Power of 2 charact. sequ. } c \text{ (Ex. \ref{ex:ch_powers_2})}& \Theta(n) & \text{Not recurrent} & \Theta(\log n) \\[1.5pt]
    (5,3)\text{-Toeplitz word } z \text{ (Ex. \ref{ex:Toeplitz})}& \Theta(n^{\frac{\log 5}{\log 5-\log 3}}) & \text{Uniformly recurrent} & \text{Not constant}\\[1.5pt]
    \end{array}
    $
  \end{center}

    \caption{Factor complexity function $p_x$, recurrence, and string attractor profile function $s_x$ for some infinite words.}
    \label{fig:my_label}
\end{figure}

Finally, we pose the problem of what values the string attractor profile function can assume, and in particular, whether an upper bound exists for these values.
We therefore prove the following proposition.

\begin{proposition}
\label{prop:lz76}
Let $x$ be an infinite word. Then $s_x(n)=O(\frac{n}{\log n})$.
\end{proposition}

\begin{proof}
The proposition can be proved by combining results from \cite{KempaP18} and \cite{lempel_ziv76}. In fact, in \cite{KempaP18} it has been proved that, for a given finite word, there exists a string attractor of size equal to the number $z$ of phrases of its LZ77 parsing. In \cite{lempel_ziv76} it has been proved that an upper bound on $z$ for a word of length $n$ is  $\frac{n}{(1-\epsilon_n)\log_{\sigma} n}$, where $\epsilon_n=2\frac{1+\log_{\sigma}(\log_{\sigma} (\sigma n))}{\log_{\sigma} n}$ and $\sigma$ is the size of the alphabet.
\qed
\end{proof}

We wonder if the bound of Proposition \ref{prop:lz76} is tight, i.e. if there exists an infinite word $x$ such that $s_x=\Theta(\frac{n}{\log n})$ for each $n\geq n_0$, for some positive $n_0$. Certainly, it is possible to construct an infinite word $x$ for which there exists a sub-sequence of positive integers $n_i$, for $i>0$, such that $s_x(n_i)=\Theta(\frac{n_i}{\log n_i})$. For instance, such a word $x$ can be constructed by using a suitable sequence of de Brujin words. However, having information about the values of the string attractor profile function on a sub-sequence $n_i$ does not allow us to determine its behavior for the remaining values of $n$.


\section{String attractor profile function on purely morphic words}
In this section, we consider the behavior of string attractor profile function for an infinite word $x$, when it is a fixed point of a morphism. Note that morphisms represents an interesting mechanism to generate infinite families of repetitive sequences, which has many mathematical properties (\cite{allouche_shallit_2003,durand_perrin_2022,beal_perrin_restivo}). Some repetitiveness measures have been explored when applied to words $x$ generated by morphisms. In \cite{FrosiniMRRS_DLT22} the number $r$ of BWT equal-letter runs has been studied for all prefixes obtained by iterating a morphism.  In \cite{IlieC2007} the measure $z_x(n)$ that gives the number $z$ of phrases in the LZ77 parsing of the prefix $x[1,n]$ has been studied. It has been proved that both $z$ and $r$ are upper bound for the measure $\gamma^*$, when they are applied to finite words. The bounds on the measure $z$ proved in \cite{IlieC2007} can be used to prove the following theorem.

\begin{theorem}
\label{th:sx_O(i)_morphism}
Let $x=\varphi^\infty(a)$ be the fixed point of a morphism $\varphi$ prolongable on $a\in \Sigma$.
Then, $s_x(n) = O(i)$, where $i$ is such that $|\varphi^i(a)| \leq n < |\varphi^{i+1}(a)|$.
\end{theorem}

To prove Theorem \ref{th:sx_O(i)_morphism}, we use a known results on another measure of repetitiveness over purely morphic word.

The \emph{LZ-parsing} of a word $w$ is its factorization $LZ(w) = v_1 v_2\cdots v_z$ built left to right in a greedy way by the following rule: each new factor (also called an LZ-phrase) $v_i$ is either the leftmost occurrence of a letter in $w$ or the longest prefix of $v_i \cdots v_z$ which occurs, as a factor, in $v_1 \cdots v_{i-1}$.
For an infinite word $x$, the \emph{LZ complexity} $z_x$ counts for each $n>0$ the size of the LZ-parsing for the prefix of $x$ of length $n$, that is $z_x(n) = LZ(x[1,n])$. 
In \cite{IlieC2007} a tight bound for the LZ complexity measure of purely morphic words is given, as reported in the following proposition.

\begin{proposition}[\cite{IlieC2007}]
\label{prop:IlieC}
For a nonerasing morphism $\varphi$ that admits the fixed point $\varphi^\infty(a)$, $z(\varphi^i(a))$ is either $\Theta(1)$, if $\varphi^\infty(a)$ is ultimately periodic, or $\Theta(i)$, otherwise.
\end{proposition}

\begin{proof}[Theorem \ref{th:sx_O(i)_morphism}]
A well known property of the measure $z$ is that it is monotone with respect to the concatenation, that is $z(u) \leq z(uv)$ for all $u,v \in \Sigma^*$.
Let $n_i = |\varphi^i(a)|$.
By Proposition \ref{prop:IlieC}, for every prefix $x[1,n]$ of the fixed point such that $|\varphi^i(a)| \leq n < |\varphi^{i+1}(a)|$ for some $i >0$, there exist two constant $c_1, c_2>0$ such that $c_1\cdot i \leq z(x[1,n_i]) \leq z(x[1,n]) \leq z(x[1,n_{i+1}]) \leq c_2\cdot  i + c_2$, that is $z(x[1,n]) = \Theta(i)$.
Finally, in \cite{KempaP18} it has been proved that for every word $w \in \Sigma^*$, it holds that $\gamma^*(w) \leq z(w)$ and the thesis follows.
\qed
\end{proof}

In the following, we provide a finer result in the case of binary purely morphic word.

\begin{theorem}
\label{th:binary_morph}
Let $x=\mu^\infty(a)$ be the binary fixed-point of a morphism $\mu:\{a,b\}^* \mapsto \{a,b\}^*$ prolongable on $a$.
Then, either $s_x(n) = \Theta(1)$ or $s_x(n) = \Theta(\log n)$, and it is decidable when the first or the latter occurs.
\end{theorem}

\begin{proof}
If $x$ is ultimately periodic, then by Proposition \ref{prop:ult_per_sx_bounded} follows that $s_x(n) = \Theta(1)$.
Suppose now $x$ is aperiodic.
For morphisms defined on a binary alphabet, it holds that if $x= \mu^\infty(a)$ is aperiodic, then $|\mu^i(a)|$ grows exponentially with respect to $i$ (see  \cite{FrosiniMRRS_DLT22}). 
Moreover, if $\mu$ is primitive, then by  \cite[Theorem 1]{DAMANIK2006766} and \cite[Theorem 10.9.4]{allouche_shallit_2003} $x$ is linearly recurrent, and by Theorem \ref{th:shallit-schaeffer_lin_rec} we have that $s_x(n) = \Theta(1)$.
If $\mu$ is not primitive, as summed up in \cite{FrosiniMRRS_DLT22}, then only one of the following cases occurs: (1) there exist a coding $\tau:\Sigma \mapsto \{a,b\}^+$ and a primitive morphism $\varphi: \Sigma^* \mapsto \Sigma^*$ such that $x = \mu^\infty(a) = \tau(\varphi^\infty(a))$ \cite{DBLP:conf/icalp/Pansiot84}; (2) $x$ contains arbitrarly large factors on $\{b\}^*$. 
For case (1), since $\tau$ preserves the recurrence of a word and that $\varphi^\infty(a)$ is linearly recurrent, then $x$ is linearly recurrent as well, and by Theorem \ref{th:shallit-schaeffer_lin_rec} $s_x(n) = \Theta(1)$. For case (2), one can notice that $x$ is not $\omega$-power free, and by Proposition \ref{prop:omega-powerFree} for every $k>0$ exists $n'$ such that $s_x(n)>k$, for every $n \geq n'$.
More in detail, the number of distinct maximal runs of $b$'s grows logarithmically with respect to the length of the prefixes of $x$ \cite{FrosiniMRRS_DLT22}, i.e. $s_x(n) = \Omega(\log n)$.
On the other hand, by Theorem \ref{th:sx_O(i)_morphism} we know that $s_x(n) = O(i)$, where $i>0$ is such that $|\mu^i(a)| \leq n < |\mu^{i+1}(a)|$.
Since $i = \Theta(\log n)$, we can further deduce an upper bound for the string attractor profile function and it follows that $s_x(n) = \Theta(\log n)$.
Finally, from a classification in \cite{FrosiniMRRS_DLT22} we can decide, only from $\mu$, if either $s_x(n)= \Theta(1)$ or $s_x(n)= \Theta(\log n)$. 
\qed
\end{proof}

Note that the result of Theorem \ref{th:binary_morph} does not contradict a possible positive answer to the Questions  \ref{conj:sx_bounded} and \ref{conj:sx_bounded_omega}, because the infinite words $x$ with linear factor complexity and such that $s_x(n)=\Theta(\log n)$ are not $\omega$-power free. Moreover, the same bounds of Theorem \ref{th:binary_morph} have been obtained for a related class of words, i.e. the automatic sequences, as reported in the following theorem.
In short, an infinite word $x$ is \emph{$k$-automatic} if and only if there exists a coding $\tau: \Sigma \mapsto \Sigma$ and a $k$-uniform morphism $\mu_k$ such that $x = \tau(\mu_k^\infty(a))$, for some $a \in \Sigma$ (\cite{allouche_shallit_2003}).  

\begin{theorem}[\cite{schaeffer_shallit2021string}]
\label{th:shallit-schaeffer}
Let $x$ be a $k$-automatic infinite word. Then, either $s_x(n)= \Theta(1)$ or $s_x(n) = \Theta(\log n)$, and it is decidable when the first or the latter occurs.   
\end{theorem}


\section{New string attractor-based complexities}

In this section we introduce two new string attractor-based complexity measures, called \emph{span complexity} and \emph{leftmost complexity}, that allow us to obtain a finer classification for infinite families of words that takes into account the distribution of positions in a string attractor of each prefix of an infinite word. Examples \ref{ex:period_doubling} and \ref{ex:Fibonacci} show two infinite words, Period-Doubling word and Fibonacci word, which are not distinguishable if we consider their respective string attractor profile function. In fact, they are point by point equal to $2$, definitively. The situation is very different if we look at how the positions within a string attractor are arranged.


\paragraph{Span and leftmost string attractor of a finite word}
Let $w$ be a a finite word and let $\calG$ be set of all string attractors $\Gamma=\{\delta_1,\delta_2, \ldots,\delta_{\gamma}\}$ for $w$, with $\delta_1 < \delta_2 < \ldots < \delta_{\gamma}$ for any $1 \leq \gamma \leq |w|$. We define \emph{span} of a finite word the value $\sp(w)=\min_{\Gamma\in \calG} \{\delta_{\gamma}-\delta_1\}$.
In other words, $\sp(w)$ gives the minimum distance between the rightmost and the leftmost positions of any string attractor for $w$. Moreover, given two string attractors $\Gamma_1$ and $\Gamma_2$, we say that $\Gamma_1$ is more to the left of $\Gamma_2$ if the rightmost position of $\Gamma_1$ is smaller than the rightmost position of $\Gamma_2$. Then, we define $\lm(w) = \min_{\Gamma\in \mathcal{G}}\{\delta_{\gamma}\in \Gamma\}$. Any $\Gamma\in \calG$ such that $\delta_\gamma = \lm(w)$ is called \emph{leftmost string attractor} for $w$.

\begin{example}
\label{ex:span_lm}
Let us consider the word $w = \overline{ab}c\underline{\overline{c}ab}c$. One can see that the sets $\Gamma_1= \{4,5,6\}$ (underlined positions) and $\Gamma_2 =\{1,2,4\}$ (overlined positions) are two suitable string attractors for $w$.
Even if both string attractors are of smallest size ($|\Gamma_1| = |\Gamma_2| = |\Sigma|)$, only the set $\Gamma_1$ is of minimum span, since all of its positions are consecutive, and therefore $\sp(w) = 6-4 = 2$.
On the other hand, one can see that $\max\{\Gamma_2\} < \max\{\Gamma_1\}$.
Moreover, one can notice that the set $\Delta = \{1,2,3\}$ is not a string attractor for $w$, and therefore $\lm(w) = \max\{\Gamma_2\} = 4$.
\end{example}

Example \ref{ex:span_lm} shows that for a finite word $w$, these two measures can be obtained by distinct string attractors.
In fact, the set $\{2,3,4\}$ is not a string attractor for $w = abccabc$, hence it does not exists $\Gamma'(w) = \{\delta_1, \delta_2, \ldots, \delta_{\gamma'}\}\in \calG$ such that $\delta_\gamma=4$ and $\delta_{\gamma'}-\delta_1 = 2$. 

The value $\sp(w)$ can be used to derive an upper-bound on the number of distinct factors of $w$, as shown in the following lemma. Such a result will be used to find upper bounds on the factor complexity of an infinite word.

\begin{lemma}
\label{le:bound_factors_span}
Let $w$ be a finite word.
Then, for all $0 < n \leq |w|$, it holds that $|F(w) \cap \Sigma^n| \leq n+\sp(w)$.
\end{lemma}

\begin{proof}
Let $\Gamma = \{\delta_1, \delta_2, \ldots, \delta_{\gamma}\}$ be a string attractor for $w$ such that $\delta_{\gamma} - \delta_1 = \sp(w)$. 
Then, the superset $X= \{i\in \mathbb{N}\mid \delta_1\leq i\leq \delta_{\gamma}\}$ of $\Gamma$ is a string attractor for $w$ as well.
Since every factor has an occurrence crossing a position in $X$, it is possible to find all factors in $F(w) \cap \Sigma^n$ by considering a sliding window of length $n$, starting at position $\max\{\delta_1 - n + 1, 1\}$ and ending at $\min\{\delta_{\gamma}, |w|-n +1\}$.
One can see that this interval is of size at most $\delta_{\gamma} - (\delta_1 - n + 1) + 1 = \delta_{\gamma} - \delta_1 + n =n+ \sp(w) $ and the thesis follows.
\qed
\end{proof}

The following proposition shows upper bounds for the measures $\gamma^*$, $\sp$ and $\lm$, when a morphism is applied to a finite word $w$.

\begin{proposition}
\label{prop:sa_complx_morph}
Let  $\varphi:  \Sigma^* \mapsto  \Sigma'^*$  be a morphism.
Then, there exists $K>0$ which depends only from $\varphi$ such that, for every $w \in \Sigma^+$, it holds that:
\begin{enumerate}
    \item $\gamma^*( \varphi (w) )  \leq 2\gamma^*(w) + K$;
    \item $\sp( \varphi (w) )  \leq K\cdot \sp(w)$;
    \item $\lm( \varphi (w) )  \leq K \cdot \lm(w)$.
\end{enumerate}
\end{proposition}

\begin{proof}
Consider any string attractor $\Gamma(w) = \{j_1, j_2, \ldots, j_{\gamma}\}$ for $w$.
We now show how to build a valid string attractor for $\varphi(w)$ starting from $\Gamma(w)$ in two steps.
First we consider the set of factors of the images of symbols, that is $F_{\varphi} = \bigcup_{a\in \Sigma}F(\varphi(a))$.
Recall that for every symbol $a \in \Sigma$ there is at least a position $j_k \in \Gamma(w)$ such that $w_{j_k} = a$.
Then, for every $a \in \Sigma$ we can choose any suitable string attractor $\Gamma(\varphi(a))$ and overlay it on the occurrence $\varphi(w_{j_k})$ to cover the factors in $F(\varphi(a))$.
Hence, every factor in $F_\varphi$ crosses at least a position in $$\calT^\varphi = \bigcup_{a\in\Sigma}\{|\varphi(w[1,j_{k-1}])| + \delta \mid \delta \in \Gamma(\varphi(a)) \text{ and }j_k \in \Gamma(w) \text{ s.t. } w_{j_k} = a\}.$$ 

Let us consider $\calF = \{u \in F(\varphi(w)) \mid u \notin F_{\varphi}\}$ as the set of factors of $\varphi(w)$ that are not factors of any image over a symbol in $\Sigma$, which are those left to cover.
Then there exist $v, u_1, u_2 \in \Sigma^*$ such that $u = u_1 \varphi(v) u_2$, where $v\in F(w)$, and $u_1$ and $u_2$ are respectively a proper prefix of $\varphi(a)$ and a proper suffix of $\varphi(b)$, for some $a,b\in \Sigma$.
Let $\mathcal{T}^f = \{|\varphi(w[1,j_{k}-1])| +1 \mid j_k \in \Gamma(w)\}$ be the set of positions in correspondence to the first symbol of $\varphi(w_{j_k})$, where $j_k$ is a position in $\Gamma(w)$.
Analogously, we define the set $\mathcal{T}^l = \{|\varphi(w[1,j_k])| \mid j_k \in \Gamma(w)\}$ as the set of positions in correspondence to the last symbol of $\varphi(w_{j_k})$.
Note that by construction, if $v$ is a factor of $w$, then exists an occurrence of $\varphi(v)$ that crosses a position either in $\mathcal{T}^f$ or in $\mathcal{T}^l$.
Finally, consider those factors in $\calF$ that have either or both $u_1,u_2 \in \Sigma^+$.
Then, there are $a, b \in \Sigma$ such that $avb \in F(w)$, and therefore either there exists an occurrence of $avb$ in $w$ such that a string attractor falls within $v$ (as for the previous case), or an attractor position falls either in $a$ or $b$.
In both cases, $\calT^f \cup \calT^l$ covers all these factors.
Hence, the set $\Gamma'(\varphi(w)) = \calT^f \cup \calT^{l} \cup \calT^{\varphi} = \{\delta_1, \delta_2, \ldots, \delta_{\gamma'}\}$ is a string attractor for $\varphi(w)$.
Let $\ell = \max_{a\in \Sigma}\{|\varphi(a)|\}$, that is the longest image for a symbol in $\Sigma$.
One can notice that:
\begin{enumerate}
    \item if $|\Gamma(w)| = \gamma^*(w)$, then
    \begin{equation*}
        \begin{split}
            \gamma^*(\varphi(w)) & \leq|\Gamma'(\varphi(w))| \leq |\calT^{f}| + |\calT^{l}| + |\calT^{\varphi}|\\
            & \leq 2 \gamma^*(w) + \sum_{a \in \Sigma}|\Gamma(\varphi(a))|;
        \end{split}
    \end{equation*}
    \item if $j_{\gamma} - j_1 = \sp(w)$, then by construction we have $\delta_1 = |\varphi(w[1, j_1 -1])|+1\in \calT^f$ and $\delta_{\gamma'} = |\varphi(w[1, j_{\gamma}])| \in\calT^l$, and therefore 
    \begin{equation*}
        \begin{split}
            \sp(\varphi(w)) & \leq \delta_{\gamma'}- \delta_1 = |\varphi(w[1, j_{\gamma}])| - (|\varphi(w[1, j_1 -1])| + 1) \\
            & = |\varphi(w[j_1,j_{\gamma}])| - 1 \\
            & \leq \ell (\sp(w)+1);
        \end{split}
    \end{equation*}
    \item if $j_\gamma = \lm(w)$, then 
    \begin{equation*}
        \begin{split}
            \lm(\varphi(w)) & \leq \delta_{\gamma'} = |\varphi(w[1, j_{\gamma}])|\\
            & \leq \ell \cdot \lm(w).
        \end{split}
    \end{equation*}
\end{enumerate}
Finally, we can choose $K = \ell \cdot |\Sigma|$ and the thesis follows for all three cases.
Note that $K$ is independent from $w$.
\qed
\end{proof}


\paragraph{Span complexity and leftmost complexity}
The following measures take into account the distribution of the positions within a string attractor for each prefix of an infinite word $x$.  More in detail, we define the \emph{string attractor span complexity} (or simply \emph{span complexity}) of an infinite word $x$ as $\spf_x(n) = \sp(x[1,n])$. We also introduce the \emph{string attractor leftmost complexity} (or simply \emph{leftmost complexity}) of an infinite word $x$, defined as $\lm_x(n) = \lm(x[1,n])$. Example \ref{ex:period_doubling} shows the behaviour of such measures when the period-doubling word is considered. Proposition \ref{le:string_attractor_based_compl_relationship} shows the relationship between the measures $s_x$, $\sp_x$ and $\lm_x$. 

\begin{example}
\label{ex:period_doubling}
Let us consider the period-doubling sequence $$pd = 101110101011\cdots,$$ that is the fixed point of the morphism $1 \mapsto 10$, $0 \mapsto 11$.
It has been proved in \cite{schaeffer_shallit2021string} that $s_{pd}(n)=2$ for any $n>1$, while $\spf_x(n) = 1$ when $1<n\leq 5$, and $\spf_x(n) = 2^i$ when $3\cdot 2^i \leq n < 3\cdot 2^{i+1}$ and $i\geq 1$.
\end{example}

\begin{proposition}
\label{le:string_attractor_based_compl_relationship}
Let $x$ be an infinite word.
Then, $$s_x(n)-1\leq \spf_x(n)\leq \lm_x(n).$$
\end{proposition}

\begin{proof}
Let $\Gamma=\{\delta_1, \delta_2, \ldots, \delta_{\gamma}\}$ be a leftmost string attractor, i.e. $\delta_{\gamma} = \lm_x(n)$.
It is possible to check that $\lm_x(n) = \delta_{\gamma} \geq \delta_{\gamma} - \delta_1 \geq \spf_x(n)$.
Let $\calG$ be the set of all string attractors for $x[1..n]$ and let $\Gamma' = \{\lambda_1, \ldots, \lambda_{\gamma'}\}\in \calG$ be a string attractor such that $\lambda_{\gamma'}-\lambda_1 = \spf_x(n)$.
Recall that, for every string attractor $\Gamma' \in \calG$ and a set $X$ such that $\Gamma' \subseteq X$, it holds that $X\in \calG$ as well.
Then, the set $X = \{i \in \mathbb{N} \mid \lambda_1 \leq i \leq \lambda_{\gamma'}\}$ is a string attractor for $x[1..n]$.
Finally, $\spf_x(n) = \lambda_{\gamma'} - \lambda_1 = |X| -1  \geq s_x(n) -1$ and the thesis follows.
\qed
\end{proof}

The following two propositions show that the boundedness of the two new complexity measures here introduced can be related to some properties of repetitiveness for infinite words, such as periodicity and recurrence.

\begin{proposition}
\label{prop:mx_ult_per_word}
Let $x$ be an infinite word.
$x$ is ultimately periodic if and only if there exists $k > 0$ such that $\lm_x(n)\leq k$, for infinitely many $n>0$.
\end{proposition}

\begin{proof}
First we prove the first implication.
If $x$ is ultimately periodic, then there exist $u\in \Sigma^*$ and $v\in \Sigma^+$ such that $x = uv^\omega$.
Observe that, for any $n\geq|uv|$, the set $\Gamma = \{ i \in \mathbb{N} \mid 1\leq i \leq |uv|\}$ is a string attractor for $x[1..n]$, since every factor that starts in $uv$ is clearly covered, and every factor that lies within two or more consecutive $v$'s has another occurrence starting in the first $v$.
It follows that we can pick $k = |uv|$ such that $\lm_x(n)\leq k$ for every $n>0$.

We now show the other direction of the implication.
By hypothesis, for all $n > 0$ there exists $n'\geq n$ and a set $\Gamma'$ such that $\Gamma' = \{\delta_1, \delta_2, \ldots, \delta_{\gamma'}\}$ is a string attractor for $x[1..n']$, with $\delta_1 < \delta_2 < \ldots < \delta_{\gamma'} \leq k$.
Hence, also the superset $\Gamma'' = \{ i \in \mathbb{N} \mid 1\leq i \leq \min\{n',k\}\}$ is a string attractor for $x[1,n']$.
One can notice that, for each $n'>k$, $\Gamma''$ can capture at most $k$ distinct factors of length $n$, i.e. one factor starting at each position of $\Gamma''$.
Therefore, for all $n'>k$ we have that $p_x(n') \leq k = \Theta(1)$.
The thesis follows from the monotonicity of the factor complexity with respect to the concatenation.
\qed
\end{proof}

\begin{proposition}\label{prop:bounded_span}
Let $x$ be an infinite word.
If there exists $k>0$ such that $\spf_x(n) \leq k$ for infinitely many $n$, then $x$ is ultimately periodic or recurrent.
\end{proposition}

\begin{proof}
Let $x$ be an ultimately periodic word.
By Propositions \ref{prop:mx_ult_per_word} and \ref{le:string_attractor_based_compl_relationship}, there exists $k>0$ such that $k \geq \lm_x(n) \geq \spf_x(n)$, for every $n>0$.
So, let us suppose that $x$ is aperiodic, and by contraposition assume that $x$ is not recurrent.
Then, for a sufficiently large value $n'$, there exists a factor $u$ of $x[1,n']$ that occurs exactly once in $x$. 
It follows that in order to cover the factor $u$, any suitable string attractor $\Gamma(x[1,n])$ with $n > n'$ must have its first position $\delta_1 \leq n'$ .
Let us consider then all the prefixes of length $n > n'$ (ignoring a finite set does not affect the correctness of the proof).
From Proposition \ref{prop:mx_ult_per_word} one can observe that $x$ being aperiodic implies that, for each $k>0$, we can find only a finite number of $n>0$ such that $k > \lm_x(n)$.
In other terms, any string attractor of a prefix of $x$ ultimately has the first position bounded above by the constant value $n'$, while $\lm_x(n)$ must grow after the concatenation of a finite number of symbols and the thesis follows.
\qed
\end{proof}


\section{Words with constant span complexity}\label{sec:constant_span}

In this section, we consider infinite words for which the \spf\  complexity measure assumes a constant value for infinite points. By using Proposition \ref{prop:bounded_span}, we know that, under this assumption, an infinite word is either ultimately periodic or recurrent. In this section we focus our attention on aperiodic words by showing that different constant values for the span complexity characterize different infinite families of words.


\paragraph{Sturmian words}
They are very well-known combinatorial objects having a large number of mathematical properties and characterizations. Sturmian words have also a geometric description as digitized straight lines \cite[Chapt.2]{Loth2}. Among aperiodic binary infinite words, they are those with minimal factor complexity, i.e. an infinite word $x$ is a \emph{Sturmian word} if $p_x(n)=n+1$, for $n\geq 0$. Moreover, Sturmian words are uniformly recurrent. An important class of Sturmian words is that of \emph{Characteristic Sturmian words}. A Sturmian word $x$ is \emph{characteristic} if both $0x$ and $1x$ are Sturmian words. An important property of characteristic Sturmian words is that they can be constructed by using finite words, called \emph{standard Sturmian words}, defined recursively as follows. Given an infinite sequence of integers $(d_0,d_1,d_2, \ldots)$, with $d_0\geq 0, d_i>0$ for all $i>0$, called a {\em directive sequence}, $x_0 = b, x_1 = a, x_{i+1} = x_i^{d_{i-1}}x_{i-1}$, for $i\geq 1$. A characteristic Sturmian word is the limit of a infinite sequence of standard Sturmian words, i.e. $x=\lim_{i\mapsto \infty}x_i$. Standard Sturmian words are finite words with many interesting combinatorial properties and appear as extremal case for several algorithms and data structures 
\cite{KnuthMorrisPratt1977,CastiglioneRS08,CS09,SciortinoZ07,MaReSc}. 

\smallskip
The following theorem shows that each prefix of a characteristic Sturmian word has a smallest string attractor of $\sp$ $1$, i.e. consisting of two consecutive positions.

\begin{theorem}
\label{th-standard_profile}
Let $x$ be a characteristic Sturmian word and let $x_0, x_1, \ldots, x_k, \ldots$ be the sequence of standard Sturmian words such that $x=\lim_{k\to\infty}x_k$.
Let $\overline{n}$ be the smallest positive integer such that $\alph(x[1..\overline{n}])=2$.
Then, $s_x(n)=2$ and $\sp_x(n) = 1$ for $n\geq \overline{n}$. In particular, a string attractor for $x[1,n]$ is given by 
$$\Gamma_n =
\begin{cases}
    \{1\}, & \mbox{if } n <\overline{n};\\
    \{|x_{k'-1}|-1, |x_{k'-1}|\}, & \mbox{if }|x_{k'}|\leq n \leq |x_{k'}|+|x_{k'-1}|-2;\\
    \{|x_{k'}|-1,|x_{k'}|\}, & \mbox{if }|x_{k'}|+|x_{k'-1}|-1\leq n < |x_{k'+1}|
\end{cases}
$$
where $k'$ is the greatest integer $k\geq 2$ such that $x_k$ (with $|x_k|\geq \overline{n}$) is prefix of $x[1,n]$.
Moreover, $\Gamma_n$ is the leftmost string attractor of $x[1,n]$.
\end{theorem}

\begin{proof}
Let $q_0, q_1, \ldots, q_k, \ldots$ be the directive sequence of the standard Sturmian words $\{x_k\}_{k\geq 0}$.
We can suppose that $q_0>0$.
The proof can be given by induction.
It is easy to see that the statement is true for $x_2=a^{q_0}b$, since the set $\{|x_2|-1,|x_2|\} $ is a string attractor. 
Note that this set is the leftmost string attractor consisting of consecutive positions.
Moreover, such two positions are a string attractor for $x_3=(a^{q_0}b)^{q_{1}}a$ and for each of its prefixes, too.
It is known that, for $i\geq 1$, $x_{2i}=C_{2i}ab$ and $x_{2i+1}=C_{2i+1}ba$, where $C_{2i}$ and $C_{2i+1}$ are palindrome words.
Let us consider the case (i) $q_{2i-1}=1$. In this case $x_{2i+1}=x_{2i}x_{2i-1}=C_{2i}\overline{ab}C_{2i-1}ba$.
From \cite[Theorem 22]{MRRRS_TCS21}, it follows that $\{|x_{2i}|-1,|x_{2i}|\}$ (the overlined positions) is a string attractor for $x_{2i+1}$ and it is the leftmost string attractor consisting of consecutive positions.
In fact, $bC_{2i-1}b$ occurs only once in $x_{2i+1}$ because $|C_{2i}|+2$ and $|C_{2i-1}|+2$ are periods of $C_{2i+1}$. 
Let us consider the subcase (i.1) in which $q_{2i}=1$. In this case $x_{2i+2}=x_{2i+1}x_{2i}=C_{2i+1}\overline{ba}C_{2i}ab= C_{2i}\underline{ab}C_{2i-1}\overline{ba}C_{2i}ab$.
From \cite[Theorem 22]{MRRRS_TCS21}, it follows that $\{|x_{2i+1}|-1,|x_{2i+1}|\}$ (the overlined positions) is a string attractor for $x_{2i+2}$. Since $|C_{2i}|+2$ is a period of $C_{2i+2}$, then the set $\{|x_{2i}|-1,|x_{2i}|\}$ (the underlined positions) is a string attractor for the prefix of length $|x_{2i+1}|+|x_{2i}|-2$.
When the prefix of length $|x_{2i+1}|+|x_{2i}|-1$ is considered, then $\{|x_{2i+1}|-1,|x_{2i+1}|\}$ is a string attractor since $aC_{2i}a$ occurs only once in $x_{2i+2}$.
Let us consider the subcase (i.2) in which $q_{2i}>1$.
In this case $x_{2i+2}=(x_{2i+1})^{q_{2i}}x_{2i}=C_{2i+1}\overline{ba}C_{2i+1}ba\cdots C_{2i+1}baC_{2i}ab$. 
From \cite[Theorem 22]{MRRRS_TCS21}, it follows that $\{|x_{2i+1}|-1,|x_{2i+1}|\}$ (the overlined positions) is a string attractor for $x_{2i+2}$.
On the other hand $x_{2i+2}=C_{2i}\underline{ab}C_{2i-1}\overline{ba}C_{2i}abC_{2i-1}ba\cdots C_{2i+1}baC_{2i}ab$ and we know that $aC_{2i}a$ does not occur in the prefix of length $|x_{2i+1}|+|x_{2i}|-2$.
This means that for that prefix $\{|x_{2i}|-1,|x_{2i}|\}$ (the underlined positions) is a string attractor. 
Let us consider the case (ii) $q_{2i-1}>1$. 
In this case $x_{2i+1}=x_{2i}^{q_{2i-1}}x_{2i-1}=C_{2i}\overline{ab}C_{2i}ab\cdots C_{2i}abC_{2i-1}ba$. 
From \cite[Theorem 22]{MRRRS_TCS21}, it follows that $\{|x_{2i}|-1,|x_{2i}|\}$ (the overlined positions) is a string attractor for $x_{2i+1}$ and it is the leftmost string attractor consisting of consecutive positions.
In fact, $C_{2i-1}b$ is a prefix of $C_{2i}ab$ and $bC_{2i-1}b$ occurs only once in $x_{2i+1}$ because $|C_{2i-1}|+2$ is period of $C_{2i+1}$.
Let us consider the subcase (ii.1) in which $q_{2i}=1$. 
In this case $x_{2i+2}=x_{2i+1}x_{2i}=C_{2i+1}\overline{ba}C_{2i}ab= C_{2i}\underline{ab}\cdots C_{2i}abC_{2i-1} \overline{ba}C_{2i}ab$.
From \cite[Theorem 22]{MRRRS_TCS21}, it follows that the set $\{|x_{2i+1}|-1,|x_{2i+1}|\}$ (the overlined positions) is a string attractor for $x_{2i+2}$.
Since $|C_{2i}|+2$ is a period of $C_{2i+2}$, then the set $\{|x_{2i}|-1,|x_{2i}|\}$ (the underlined positions) is a string attractor for the prefix of length $|x_{2i+1}|+|x_{2i}|-2$.
When the prefix of length $|x_{2i+1}|+|x_{2i}|-1$ is considered, then $\{|x_{2i+1}|-1,|x_{2i+1}|\}$ is a string attractor since $aC_{2i}a$ occurs only once in $x_{2i+2}$.
Let us consider the subcase (i.2) in which $q_{2i}>1$. In this case $x_{2i+2}=(x_{2i+1})^{q_{2i}}x_{2i}=C_{2i+1}\overline{ba}C_{2i+1}ba\cdots C_{2i+1}baC_{2i}ab$.  From \cite[Theorem 22]{MRRRS_TCS21}, it follows that $\{|x_{2i+1}|-2,|x_{2i+1}|-1\}$ (the overlined positions) is a string attractor for $x_{2i+2}$. 
On the other hand $x_{2i+2}=C_{2i}\underline{ab}\cdots C_{2i}abC_{2i-1}\overline{ba} C_{2i+1}ba\cdots C_{2i+1}baC_{2i}ab$ and we know that $C_{2i}a$ is a prefix of $C_{2i+1}ba$ and $aC_{2i}a$ does not occur in the prefix of length $|x_{2i+1}|+|x_{2i}|-2$. This means that for that prefix $\{|x_{2i}|-2,|x_{2i}|-1\}$ (the underlined positions) is a string attractor. 
We can consider the cases $q_{2i}=1$ and $q_{2i}>1$ and the respective sub-cases analogously. 
\qed
\end{proof}

\begin{example}
\label{ex:Fibonacci}
Consider the infinite Fibonacci word $x = abaababaabaababaababa\ldots$ that is a characteristic Sturmian word with directive sequence $1,1,\ldots,1,\ldots$.

In Figure \ref{fig:fib_attractor} are shown the first prefixes of $x$ of length $n$ and their respective leftmost string attractor $\Gamma_n$, with $n\geq 2$.
\begin{figure}[h!]
    \begin{align*}
       \textbf{x[1]}   & =\underline{a}            & \Gamma_1 & = \{1\}\\
        \textbf{x[1,2]} & = \underline{ab}          & \Gamma_2 & = \{1,2\}\\
        \textbf{x[1,3]} & = \underline{ab}a         & \Gamma_3 & = \{1,2\}\\
        x[1,4]          & = a\underline{ba}a        & \Gamma_4 & = \{2,3\}\\
        \textbf{x[1,5]} & = a\underline{ba}ab       & \Gamma_5 & = \{2,3\}\\
        x[1,6]          & = a\underline{ba}aba      & \Gamma_6 & = \{2,3\}\\
        x[1,7]          & = aba\underline{ab}ab     & \Gamma_7 & = \{4,5\}\\
        \textbf{x[1,8]} & = aba\underline{ab}aba    & \Gamma_8 & = \{4,5\}
    \end{align*}
    \caption{Prefixes of the Fibonacci word $x$ of length up to 8 and their leftmost string attractor $\Gamma_n$. For Fibonacci words we have $\overline{n}=2$. The underlined positions in $x[1,n]$ correspond to those in $\Gamma_n$, while the prefixes in bold are standard Sturmian words.}
    \label{fig:fib_attractor}
\end{figure}
\end{example}

The following proposition shows that there is a one-to-one correspondence between each characteristic Sturmian word and the sequence of the leftmost string attractors of its prefixes.

\begin{proposition}
\label{prop:unique_leftmost_span_sturmian}
Let $x$ be a characteristic Sturmian word and, for each $n \geq 1$, let $\Gamma_n$ be the string attractor of the prefix $x[1,n]$ defined in Theorem \ref{th-standard_profile}.
Let $y$ be a characteristic Sturmian word such that $\Gamma_n$ is the leftmost string attractor for $y[1,n]$ for any $n \geq 1$.
Then, $x = y$ (up to exchanging $a$ and $b$).
\end{proposition}

\begin{proof}
Let $q_0, q_1, \ldots, q_k, \ldots$ and $p_0, p_1, \ldots, p_k, \ldots$ be the directive sequences that uniquely define $x$ and $y$ respectively.
Let $S_a$ and $S_b$ be the sets of characteristic Sturmian words that start with $a$ and $b$ respectively.
Note that $x \in S_a$ if and only if $q_0 >0$, $x \in S_b$ otherwise.
We now prove by construction that if $x,y \in S_a$ and $x \neq y$, then the sequences of leftmost string attractor with $\sp_x(n)=1$ are distinct as well. 
Consider the smallest integer $i$ such that $q_i \neq p_i$ and assume, w.l.g., that $q_i < p_i$.
If $i = 0$, then we can consider the prefix $x[1, q_0+1] = a^{q_0}b$ and it is easy to see that $\Gamma_{q_0+1}^x = \{q_0,q_0+1\}$.
However, the set $\Gamma_{q_0+1}^y = \{1\}$ is clearly the leftmost string attractor for the prefix of the same length $y[1,q_0+1] = a^{q_0+1}$.
If $i > 0$, note that $|x_k| = |y_k|$ for any $k \leq i+1$ and $|x_{i+2}| < |y_{i+2}|$. 
From Theorem \ref{th-standard_profile}, it follows that for the prefix of $x$ of length $n' = |x_{i+2}| + |x_{i+1}| - 1$ we have as leftmost string attractor the set $\Gamma_{n'}^x = \{|x_{i+2}| - 1, |x_{i+2}|\}$, but since $n' < |y_{i+2}| + |y_{i+1}| - 1$, the string $y[1,n'+1]$ admits as string attractor the set $\Gamma_{n'}^y = \{|x_{i+1}| - 1, |x_{i+1}|\} = \{|y_{i+1}| - 1, |y_{i+1}|\}$, and therefore $\Gamma_{n'}^x \neq \Gamma_{n'}^y$.
A well known property of characteristic Sturmian words is that $x \in S_a$ with directive sequence $q_0, q_1, \ldots, q_k, \ldots$, if and only if $x' \in S_b$, where $x'$ is the infinite word $x$ with $a$'s and $b$'s exchanged with directive sequence $0, q_0, \ldots, q_{k-1}, \ldots$.
Thus, the proof holds as well if $x, y \in S_b$ by considering the cases $i = 1$ and $i > 1$.
Finally, notice that by construction of $\Gamma_n$ in Theorem \ref{th-standard_profile}, the words $x$ and $x'$ generate the same sequences of standard Sturmian words, since $|x_k| = |x'_{k+1}|$ for any $k \geq 1$.
Therefore, either $y = x$ or $y = x'$.
\qed
\end{proof}

\begin{remark}
There are non-characteristic Sturmian words such that some of their prefixes do not admit any string attractor of $\sp$ $1$.
For instance, let $x=aaaaaabaaaaaabaaaaaaab\ldots$ be  the characteristic Sturmian word obtained by the directive sequence $(6,2,\ldots)$. Consider the non-characteristic Sturmian word $x'$ such that $x = aaaa\cdot x'$, hence $x'=aabaaaaaabaaaaaaab\ldots$.
Let us consider the prefix $x'[0,13] = aabaaaaaabaaaa$.
Since the $b$'s occur only at positions $2$ and $9$ and the factor $aaaaaa$ only in $x'[3,8]$, the candidates as string attractor with two consecutive positions are $\Delta_1 = \{2, 3\}$ and $\Delta_2 = \{8, 9\}$.
However, one can check that the factors $aaab$ and $baaaaa$ do not cross any position in $\Delta_1$ and $\Delta_2$ respectively.
Nonetheless, there exists a string attractor of size $2$ that does not contain two consecutive positions, that is $\Gamma = \{3,9\}$.
\end{remark}

The following theorem shows that a new characterization of Sturmian words can be obtained in terms of \sp\ of the prefixes. 

\begin{theorem}
\label{th:Sturmian_iff_span1}
Let $x$ be an infinite aperiodic word.
Then, $x$ is Sturmian if and only if $\sp_x(n) = 1$ for infinitely many $n>0$.
\end{theorem}

\begin{proof}Observe that for every Sturmian word $x$ has an infinite number of right special factors as prefixes, as for every aperiodic and uniformly recurrent word.
Moreover, for every right special factor $v$ of a Sturmian word, there is a characteristic Sturmian word $s$ with $v^R$ as prefix \cite[Proposition 2.1.23]{Loth1}.
Since for every string $v\in \Sigma^*$ and every string attractor $\Gamma(v)$ of $v$ it holds that the set $\Gamma(v^R) = \{n-i-1 \mid i\in\Gamma(v)\}$ is a suitable string attractor of $v^R$  \cite{MRRRS_TCS21}, and from Theorem \ref{th-standard_profile} we know that $\sp_s(n) = 1$ for every prefix of every characteristic Sturmian word $s$, it follows that exist infinite prefixes $v$ of $x$ such that $\sp_x(|v|) = 1$, that is our thesis.

For the other direction of the implication, recall that an infinite word $x$ is aperiodic if and only if $p_x(k) \geq k+1$ for all $k>0$.
Moreover, by hypothesis for every $n > 0$ exists $n'>n$ such that $\sp_x(n') = 1$. 
It follows that $|F(x[1,n])\cap \Sigma^k| \leq |F(x[1,n'])\cap \Sigma^k| \leq n + \spf_x(n') = n + 1$ for every $n > 0$, and therefore $x$ is Sturmian.
\qed
\end{proof}


\paragraph{Quasi-Sturmian words}
Let us consider now the \emph{quasi-Sturmian} words, defined in \cite{Cassaigne97} as follows: a word $x$ is \emph{quasi-Sturmian} if there exist integers $d$ and $n_0$ such that $p_x(n)=n+d$, for each $n\geq n_0$.
The infinite words having factor complexity $n+d$ have been also studied in \cite{Heinis2004} where they are called \vir{words with minimal block growth}. Quasi-Sturmian words can be considered the simplest generalizations of Sturmian words in terms of factor complexity. In \cite{Cassaigne97} the following characterization of quasi-Sturmian words has been given. 

\begin{theorem}[\cite{Cassaigne97}]
An infinite word $x$ over the alphabet $\Sigma$ is quasi-Sturmian if and only if it can be written as $x=w\varphi(y)$, where $w$ is a finite word and $y$ is a Sturmian word on the alphabet $\{a,b\}$, and $\varphi$ is a morphism from $\{a,b\}^*$ to $\Sigma^*$ such that $\varphi(ab)\neq \varphi(ba)$.
\end{theorem} 

The following proposition shows that constant values for the span complexity at infinitely many points imply quasi-Sturmian words, i.e., the most repetitive infinite aperiodic words after the Sturmian words.

\begin{proposition}
\label{prop_quasisturmian}
Let $x$ be an aperiodic infinite word.
If exists $k>0$ such that $\spf_x(n) \leq k$ for infinitely many $n>0$, then $x$ is quasi-Sturmian.
\end{proposition}

\begin{proof}
By hypothesis, for all $n > 0$ exists $n'\geq n$ such that $\sp_x(n') \leq k$, for some $k>0$.
Then, for every finite $n$-length prefix of $x$ and every $m>0$, by using Lemma \ref{le:bound_factors_span} it holds that $|F(x[1,n])\cap \Sigma^m| \leq |F(x[1,n'])\cap \Sigma^m| \leq m + sp_x(n') \leq m+k$.
Moreover, it is known that for every aperiodic word it holds that $p_x(n)\geq n+1$ and $p_x(n+1)> p_x(n)$, for every $n\geq0$.
Hence, there exist $k' \leq k$ and $n_0 \geq 0$ such that $p_x(n) = n+k'$, for every $n \geq n_0$ and the thesis follows.
\qed
\end{proof}

\begin{remark}
Note that, in general, the converse of Proposition \ref{prop_quasisturmian} is not true.
In fact, let $w$ be a finite word, $y$ a Sturmian word and $\varphi$ a non-periodic morphism.
Then, $x = w \varphi(y)$ is quasi-Sturmian. We can choose as finite prefix $w$ a symbol $c \notin \alph(\varphi(y))$. One can notice that in this case $x$ is not recurrent, and by Proposition \ref{prop:bounded_span} the function $\sp_x$ is not bounded by constant. Instead, if $w=\varepsilon$, then the converse of Proposition \ref{prop_quasisturmian} is true. It can be derived from Proposition \ref{prop:sa_complx_morph} and Theorem \ref{th:Sturmian_iff_span1}. 
\end{remark}

\section{Conclusions}
In this paper, we have shown that the notion of string attractor introduced in the context of Data Compression can be useful in the study of combinatorial properties of infinite words. The string attractor- based span complexity measure has indeed been used to characterize some infinite word families. The problem of characterizing words with bounded string attractor profile function remains open.  On the other hand, the two new complexity measures here introduced could be useful to represent, in a more succinct way, information on infinite sequences of words. Finally, it might be interesting to explore how the \sp\  and \lm\  measures are related to the compressor-based measures.


\bibliographystyle{splncs04}

\end{document}